\documentclass{article}
\usepackage{graphicx} % Required for inserting images
\usepackage[margin=0.75in]{geometry} %added this to make things a little easier to read
\usepackage{amsmath}
\usepackage{authblk}

\usepackage{amssymb}
\usepackage{amsthm}
\usepackage{tikz}
\newtheorem{theorem}{Theorem}[section]

\newtheorem{observation}[theorem]{Observation}
\newtheorem{remark}[theorem]{Remark}
\newtheorem{lemma}[theorem]{Lemma}
\theoremstyle{definition}
\newtheorem{definition}{Definition}[section]
\newtheorem*{problem*}{Decision Problem}
\newtheorem{example}{Example}[section]
\usepackage{scalerel} % for scalable symbols

\usetikzlibrary{automata, positioning, arrows, shapes, matrix, decorations.pathreplacing}
\tikzset{
->, % makes the edges directed
node distance=3.5cm, % specifies the minimum distance between two nodes. Change if necessary.
every state/.style={thick, fill=gray!10}, % sets the properties for each ’state’ node
initial text=$ $, % sets the text that appears on the start arrow
}

\usepackage{hyperref}
\hypersetup{
    colorlinks=true,
    linkcolor=blue,
    filecolor=magenta,      
    urlcolor=cyan
}
\title{A Numerical Approach to the Realizability Problems for Memoryless Nash and Epsilon Equilibria in Concurrent Multiplayer Reachability Games }
\author[1]{Senthil Rajasekaran}
\author[1]{Jean-François Raskin}
\author[2]{Moshe Y. Vardi}

\affil[1]{Formal Methods Group, Université Libre de Bruxelles}
\affil[2]{Department of Computer Science, Rice University}

\affil[ ]{\small \texttt{senthil.rajasekaran@ulb.be}, \texttt{jean-francois.raskin@ulb.be}, \texttt{vardi@rice.edu}}

\begin{document}

\maketitle

\section{Introduction}\label{sec:intro}

Equilibria are mathematical objects from the field of game theory~\cite{Osborne1994} that serve as solution concepts for multiplayer games. As such, there is great interest in analyzing the computational complexity of decision problems related to equilibria in multiplayer games, resulting in the field of \emph{algorithmic game theory}~\cite{roughgarden2010algorithmic}. Recently, there has also been interest in the computational complexity of equilibria from the formal methods community. There, the high-level idea is to use the computational complexity of equilibria as a proxy to measure the computational bottlenecks that arise when distributed systems are scaled. This created the relatively new subfield of \emph{equilibrium analysis}, which is also called \emph{rational verification}~\cite{principlesofmodelcheckingbook,AGHKNPSW21,rajasekaranthesis}.

Although these two fields share a similar goal, they study different settings that generally operate under different assumptions. As a subfield of formal verification, equilibrium analysis often considers \emph{state-based} games in which player actions update a global game state over time, as these state-based games represent the execution of a system with multiple global states. The players are then motivated to play the game in a way that creates a desirable sequence of observed global states. As such, players' goals can be quite complex, and are often specified using temporal logics~\cite{AGHKNPSW21}. On the other hand, there is a strong preference for deterministic settings in the equilibrium analysis literature, see~\cite{RV21,RV22,RBV23,RV26,GutierrezNPW20,iBG,bouyer2010nash,bouyer2011nash,ummels2015pure,BBMU15} for a few examples. More recently, the field has been moving towards considering probabilistic settings, see~\cite{ratverprob,RRV26a,RV25,parkerverification} for a few examples.

On the other hand, algorithmic game theory focuses almost exclusively on probabilistic non-state-based systems. As a result, algorithmic game theory has a strong preference for probabilistic settings where equilibria are mathematically guaranteed to exist. Perhaps the most famous example of this is the paper that introduced the complexity class PPAD, meant to characterize search problems in which a witness is guaranteed to exist ~\cite{DGHppad}. The introduction of PPAD has led to countless follow-up works in algorithmic game theory. Although considering state-based games in algorithmic game theory is relatively rare, when this is done, efforts are often made to add in factors that preserve the existence of equilibria, such as the use of a discount factor~\cite{JMS23}.

The goal of this paper is to study a problem setting that lies at the intersection of both fields, yet has been somewhat neglected. In this paper, we study probabilistic concurrent state-based games in which each player has a reachability goal (and, therefore, an expected probability of achieving that goal), and we consider the Nash equilibria and the $\varepsilon$-Nash equilibria, the two most widely used equilibrium concepts in the literature. Critically, we do not introduce a discount factor or a finite-time horizon, and so the existence of an equilibrium is not guaranteed. This condition allows us to meaningfully study the \emph{realizability} problem, which takes a game as input and asks whether it has an equilibrium. While there has been some work in this setting, which we outline in Section~\ref{sec:stateoftheart}, the literature is quite disorganized, making it very difficult to tell what is known. As we explain in Section~\ref{sec:stateoftheart}, the high-level answer is ``very little''.

This setting is interesting precisely because, by being state-based, probabilistic, and not guaranteed to have an equilibrium, it represents a novel setting that the literature has increasingly moved away from.  This is partly because realizability problems in this setting have proven very challenging to reason about, a point we justify in more detail in Section~\ref{sec:stateoftheart}. Even further, representing a probabilistic state-based game presents a significant challenge on its own, a point we address by introducing a novel model of transition function representation in Section~\ref{sec:prelim}. The goal of this paper is to provide a new approach of analysis in this challenging setting by bounding the size of the strategies that constitute a potential equilibrium in two ways. The first restriction we make is to consider equilibria that are made of \emph{memoryless} strategies. This is a fairly common restriction in the literature, as memoryless strategies are often valued for their beneficial mathematical properties -- c.f.~\cite{JMS23,boundedpartialinformationconcurrency,recurisveetr,stationaryepsinturnbased,etessamihandbook}, and the paper~\cite{etrstatne} should be especially noted for its relevance to this work. The second restriction concerns the representation of numbers used to model player strategies, which, to the best of our knowledge, is novel, and the main conceptual contribution of this paper. The new approach put forward in this paper is to explicitly study the representation of the strategies that constitute an equilibrium by carefully considering the field in which the numbers that constitute the strategy lie. It should be noted that~\cite{etrstatne} studies this problem without numerical restrictions, a point we discuss in much more detail in Section~\ref{sec:noell}.

First, we consider the field of rationals. This field is easy to work with and admits many beneficial complexity-theoretic properties. Nash himself, however, showed that there are games that only have irrational Nash equilibria~\cite{Nash51}, and so considering just rational numbers is not sufficient. This leaves us with an important question -- how can we reason about irrational numbers algorithmically?

We do this by considering field extensions of $\mathbb{Q}$ and working with algebraic irrational numbers symbolically, through their closed-form radical representations (irradical numbers are not considered in this step). We then show that if the field extension is represented by its basis rather than its extension elements, we can preserve the same complexity as in the rational case. As the size of the basis can be exponential in the number of the extension elements, this provides us with a curious search landscape -- if we know ``where'' to look by considering specific irrational numbers, we can solve the realizability problems quite easily. On the other hand, if we move towards generality and use a large extension, computing the basis itself becomes a significant bottleneck. This idea is more important in the analysis of $\varepsilon$-equilibria, a point we develop further in Section~\ref{sec:discussion}.  Yet this is still not enough, as there are some games that admit Nash equilibria that utilize irradical probabilities~\cite{irradicalprobsNash}. Although irradical numbers are algebraic, they cannot be represented by radical expressions and, therefore, introduce new issues of representation. Finally, to move past issues of representation, we consider the full field $\mathbb{R}$, which, unlike with $\mathbb{Q}$ and its extensions, we reason about through the use of logic -- namely, the existential theory of the reals $\exists\mathbb{R}$. This most general setting also results in the most prohibitive complexity.

When taken together, these results imply that numerical representation plays a very interesting and understated role in both the computational complexity of the realizability problems and the difficulties researchers have faced in characterizing their mathematical properties in the reachability goal setting.

\section{Preliminaries}\label{sec:prelim}

The main object of study in this paper is the probabilistic concurrent game graphs (henceforth, \emph{games}). 
\begin{definition}
A game is a 6-tuple $\mathbb{G}=\langle  Q, q_0,\Pi,A,\delta,(R_i)_{i\in \Pi} \rangle$, specified by the following components:
  \begin{itemize}
    \item A finite non-empty set of states $Q$, with $q_0 \in Q$ as the initial state.
    \item A finite non-empty set of players $\Pi$. As a matter of notational convenience, we say that there are $m$ players, i.e. $|\Pi| = m$.
    \item A finite non-empty set of action names $A$. Given the action set $A$ and the set of players $\Pi$, we define the set of \emph{joint actions} $A^{\Pi}$.  Players in a game choose actions concurrently, and a joint action is the result of all players selecting actions in $A$ concurrently.
    \item $\delta:Q \times A^{\Pi}  \rightarrow {\sf Dist}(Q)$ is the probabilistic transition relation -- 
    a function that maps states and joint actions to a distribution over states representing potential successor states.
    \item The set $R_i \subseteq Q $ is the \emph{reachability goal} of Player~$i$. 
  \end{itemize}

\end{definition}

The game starts in state $q_0 \in Q$, the initial state. Players then concurrently chose an action from the set of action names $A$. Once each player has chosen an action, the \emph{joint action}, an element of $a \in A^{\Pi}$, along with the transition function $\delta$ determines a distribution $\delta(q_0,a)$ over the next state. This distribution is then sampled in order to determine the next state. The entire process then repeats for an infinite number of rounds. This creates an infinite sequence of states called a \emph{trace} $t \in Q^{\omega}$. Since the progression of this game involves moving from state to state, we call it \emph{state-based}.
  
The transition function $\delta : Q \times A^{\Pi} \rightarrow{\sf Dist}(Q)$ is the largest part of the input of $\mathbb{G}$, as it represents a construction that is exponential in the number of players in the worst-case scenario. For this reason, the representation of $\delta$ plays a paramount role in our overall complexity-theoretic analysis. It also presents us with a unique challenge: we must unify two highly desirable properties in a single model.
\begin{enumerate}
    \item The computational complexity of all relevant queries on our transition table model should be in PTIME, so as not to overpower the overall complexity of our realizability problems.
    \item The model should be able to succinctly represent certain important subfamilies of games, especially turn-based games. The use of succinctly represented turn-based games in particular will allow us to connect our model to previous works in the literature, such as~\cite{etrstatne}.
\end{enumerate}

If $\delta$ were represented as an explicit table, it would contain $|Q| \cdot |A|^m$ rows, each of which maps to a distribution over $Q$. This explicit representation, however, does not allow for transition functions that admit succinct representations, and so we lose Property 2. On the other hand, if we were to use a Bayesian Network, this would mean that simply querying our model for a probability would be PP-hard~\cite{bayesianinferencecomplexity} (Probabilistic Polynomial Time, see~\cite{papacomplexity}), meaning that we would lose Property 1. In this paper, we model $\delta$ as a \emph{compressible table}, a new model that threads the needle and satisfies both our properties simultaneously. The compressible table allows us to capture certain succinct encodings (but not all, as then it would be equivalent to a Bayesian Network), most notably those for turn-based games, while still maintaining PTIME query complexity. 

\begin{definition}[Compressible Table]\label{compressibletabledef}
A compressible table is a table representation of $\delta$ that allows for the use of a ${\ast}$ ``don't care'' character.  For example, a row labeled by $\langle q, a{\ast}{\ast}{\ast}{\ast} \rangle$ in a 5-player game corresponds to the distribution from the state $q$ of all joint actions where player one chooses action `a', regardless of the choices of players two through five. Each row is then associated with a distribution over the state space $Q$, represented as a list mapping states to probabilities, written as rational numbers. Specifically, they are given as ratios, with both the numerator and the denominator expressed in binary.
\end{definition} 

\begin{remark} 
For the rest of the paper, when we mention representing a rational number, we mean representing its numerator and denominator in binary, with a special character separating the two. 
\end{remark} 
Given a state $q \in Q$, we denote the set of rows associated with $q$ in the compressible table $\delta$ by $R_\delta(q)$. Note that since we are allowed to use the $*$ symbol, it is possible that $|R_\delta(q)| < A^m$, where $A^m$ is the number of rows an explicit table representation would use.
\begin{definition}[The Specification of $R_\delta(q)$]
    In order to make the compressible table unambiguous, we require that there never be two rows in $R_\delta(q)$ that can encode the same joint action. For example, the rows ${\ast} a$ and $a {\ast}$ both encode the tuple `$aa$' and cannot therefore appear in the same $R_\delta(q)$. Furthermore, we ensure that $R_\delta(q)$ is total by adding a final blank row to each $R_\delta(q)$. This blank row describes the transitions for tuples that do not match another row in $R_\delta(q)$.
\end{definition}

As mentioned before, \emph{turn-based} games, in which only one player chooses an action at each state, are a highly popular model in the literature, cf.~\cite {stationaryepsinturnbased,recurisveetr}. We now show how to encode the transition function of a turn-based game using a compressible table.
\begin{remark}[Turn-Based Game Encoding]\label{rem:turn-based-encoding}
    In a turn-based game, each state $q \in Q$ ``belongs'' to at most one player. If a state belongs to no player, then the compressible table easily captures the transition information at this state by utilizing either an all ``${\ast}$'' row or the final blank row. At a state $q_1$ belonging to Player $1$, for example, the rows $a {\ast} {\ast} \ldots$ $b {\ast}  {\ast}\ldots $  capture the property that only the action choice of Player $1$ (i.e. $'a'$,'$b$', or another option ) influences the transition at state $q_1$. Note that in our formulation, all players still choose an action, but the don't-care ${\ast}$ symbol means that these choices are not taken into consideration by the transition function. 
\end{remark}

The compressible table is therefore able to capture the succinct representation of turn-based games, and we prove PTIME query complexity in the upcoming technical sections.  Furthermore, it can also capture the representation of  ``b-bounded concurrent game systems"~\cite{RV25}, a recently introduced generalization of turn-based games. Since the compressible table model considers only certain player actions at a given state, it is well-suited to the Nash equilibrium solution concept, in which players choose their actions independently. As mentioned before, however, it cannot capture every succinct representation, as this would make it essentially a Bayesian Network. There are, for example,  many important subclasses of games that the compressible table cannot model succinctly, such as the iBG~\cite{iBG}, a system in which the joint action space and the state space are the same set and the transition function is simply right projection. We stress the importance of considering the representation of the transition table in multiplayer games, and now provide the interested reader with two follow-up references on this point. For further detail on the use of Bayesian Network models in probabilistic concurrent games, see~\cite{RRV26a}. For a thorough discussion on the role of succinct representations in \emph{deterministic} games, see~\cite{RV26}.

Players are motivated to choose their actions in such a way as to satisfy their goal. The goal of Player~$i$ is to visit a state that belongs to the set $R_i \subseteq Q$. If the execution of the system results in a trace $t$ that visits $R_i$, then Player~$i$ receives a \emph{payoff} of $1$; otherwise, for traces that do not visit $R_i$, they receive a payoff of $0$. Since the game transitions probabilistically and players themselves may select their own actions probabilistically, we reason about the \emph{expected payoff} assigned to a player under some behavioral scheme. These behaviors are defined by \emph{strategies}.

\begin{definition}[Strategy for Player~$i$]\label{stratdef}
A strategy for Player~$i$ is a function $\pi_i :Q^* \rightarrow {\sf Dist}(A)$ that takes as input the observed input of the game (a history of observed states as an element of $Q^*$) and outputs a distribution over the set of action names. The set of all Player~$i$ strategies $\pi_i$ is denoted $\Pi_i$.
\end{definition}

Note that strategies consider histories in a game to be a sequence of states, not a sequence of state-action pairs. Although this distinction is meaningful in many settings, it is not in ours -- a point we justify after introducing our notions of equilibria. At a high level, this is because the focus of our paper is a setting in which player strategies are \emph{memoryless}. Memoryless strategies are also referred to as `stationary' in the literature, c.f.~\cite{etrstatne,stationaryepsinturnbased}.

\begin{definition}[Memoryless Strategy for Player~$i$]
    A memoryless strategy for Player~$i$ is a function $\pi_i : Q \rightarrow {\sf Dist}(A)$. 
\end{definition}

Thus, instead of considering the sequence of previously visited states, a memoryless strategy only takes into account the current state of the game.  Memoryless strategies are highly desirable from an algorithmic perspective due to the simplicity of their representation~\cite{etessamihandbook}.

A \emph{strategy profile} $\pi = \langle \pi_1 \ldots \pi_{m} \rangle$ is an $m$-tuple of strategies, with one strategy assigned to each player. In this paper, we are interested in strategy profiles $\pi$ such that each $\pi_i \in \pi$ is memoryless. The set of all strategy profiles is then similarly denoted $\Pi$. Given a game $\mathbb{G}$ and a strategy profile $\pi$ consisting of finite-memory strategies, each Player~$i$ is assigned an expected payoff between $0$ and $1$ that represents the probability that the set $R_i \subseteq Q$ is visited in $\mathbb{G}$ when all players follow the profile $\pi$ exactly. The expected payoff given to Player~$i$ in $\mathbb{G}$ when all players follow $\pi$ is denoted $\rho_i(\mathbb{G},\pi)$. We are interested in strategy profiles that satisfy the \emph{Nash equilibrium} solution concept.

\begin{definition}[Nash Equilibrium]
A strategy profile $\pi = \langle \pi_1 \ldots \pi_{m}\rangle$ is a Nash equilibrium if it is the case that for all players $i \in \Pi$ we have $\forall \pi^* \in \Pi_i. \rho_i(\mathbb{G},\pi) \geq \rho_i(\mathbb{G}, \langle \pi_1 \ldots \pi^* \ldots \pi_{m}\rangle )$.
\end{definition}

As before, we are interested in memoryless strategy profiles. The main focus of this paper is to find strategy profiles that are memoryless and that also satisfy the Nash equilibrium condition. Thus, we are brought to an important distinction. Although the strategy profile $\pi$ assigns a memoryless strategy $\pi_i$ to Player~$i$, Player~$i$ is not obligated to use a memoryless strategy when deviating -- he may use a general strategy of the form $\pi_i : Q^* \rightarrow {\sf Dist}(A)$. 

In addition to the Nash equilibrium, we also consider a widely-used relaxation called the $\varepsilon$-equilibrium.

\begin{definition}[$\varepsilon$-Nash Equilibrium]
A strategy profile $\pi = \langle \pi_1 \ldots \pi_{m}\rangle$ is a $\varepsilon$-Nash equilibrium for some $\varepsilon \in \mathbb{Q}$ if it is the case that for all players $i \in \Pi$ we have $\forall \pi^* \in \Pi_i. \rho_i(\mathbb{G},\pi) + \varepsilon \geq \rho_i(\mathbb{G}, \langle \pi_1 \ldots \pi^* \ldots \pi_{m}\rangle )$.
\end{definition}

Thus, the $\varepsilon$-Nash equilibrium does not recognize deviations that result in an arbitrarily small improvement, only ones that improve the expected payoff by at least some fixed constant $\varepsilon$. Sometimes, the $\varepsilon$-Nash equilibrium condition is also stated as $\forall \pi^* \in \Pi_i. (1 + \varepsilon)\rho_i(\mathbb{G},\pi) \geq \rho_i(\mathbb{G}, \langle \pi_1 \ldots \pi^* \ldots \pi_{m}\rangle )$, treating $\varepsilon$ as a multiplicative offset as opposed to an additive one. The results in this paper apply to both definitions, a point we justify in the technical sections.

While considering memoryless strategies is common in the literature~\cite{etessamihandbook}, we introduce a new restriction on strategies meant to handle issues of numerical representation. 
\begin{definition}[$\ell$-bit-representable Memoryless Strategy]
A memoryless strategy $\pi_i$ is $\ell$-bit-representable if all rational numbers used in its representation can be represented using at most $\ell$ bits.
\end{definition}

\begin{remark}
It is important to note that a deviating player is under no obligation to play a strategy that is either memoryless or $\ell$-bit-representable. Thus, in this paper, we are looking for equilibria strategy profiles $\pi$ in which each individual strategy profile $\pi_i \in \pi$ satisfies both the $\ell$-representable and memoryless conditions that cannot be profitably deviated from by \emph{arbitrary} strategies. 
\end{remark}

In this paper, we use three different numerical systems. In each setting, the phrase ``$\ell$-bit-representable'' applies only to rational numbers. First, we consider elements of  $\mathbb{Q}$ in Section~\ref{sec:Q}. In $\mathbb{Q}$, $\ell$-bit-representable means that all probabilities output by $\pi_i$ are rational numbers that can be expressed using at most $\ell$ bits. 
Then, we consider $\mathbb{Q}(r_1 \ldots r_m)$, the $d$-degree field extension obtained by extending $\mathbb{Q}$ by the radical expressions $r_1 \ldots r_m$ in Section~\ref{sec:Qext}. Here, $\ell$-bit-representable means that when we express a number in $\mathbb{Q}(r_1 \ldots r_m)$ as a linear combination of the basis elements $\{1 , b_1 \ldots b_{d-1}\}$, $a_0 + a_1b_1 + \ldots + a_{d-1}b_{d-1}$, each $a_i$ is an $\ell$-bit-representable rational number.  Finally, we remove the $\ell$-bit-representable restriction and consider general real numbers in Section~\ref{sec:noell}.

\begin{remark}\label{transitionistrepresentable}
    For $\mathbb{Q}$ and $\mathbb{Q}(r_1 \ldots r_m)$, we say that the probabilities used in the representation of $\delta$ are $t$-representable, where $t$ is the size of the largest rational number used to represent a probability in $\delta$. For $\mathbb{R}$, no such parameter $t$ is derived; but probabilities in the transition function $\delta$ are still represented as rational numbers, matching previous formulations in literature such as~\cite{stationaryepsinturnbased}   
\end{remark}

The last ingredient to introduce is the notion of \emph{payoff constraints}. For each Player~$i$, a payoff constraint $c_i = \{l_i ,u_i\}$ is an interval of values with $0 \leq l_i \leq u_i \leq 1$. Payoff constraints are widely used in the literature (see~\cite{BBMU15,CONITZER2008621} for two well-known examples), as they provide a way for a system designer to specify a range of values that they would like a player's payoff to fall in. This provides a method for the designer to filter out equilibria with properties they deem undesirable, such as those that yield minimal payoffs to each player. 

\begin{definition}[Payoff Constraints]
    Given a game $\mathbb{G}$, a constraint set $C = [c_1 \ldots c_m]$ is a set of payoff constraints $c_i = \{l_i, u_i\}$, with the constraint $c_i$ assigned to Player~$i$. A strategy profile $\pi$ satisfies a constraint set $C$ iff $\forall i \in \Pi. l_i \leq \rho_i(\mathbb{G},\pi) \leq u_i$. 
\end{definition}

Since we work with different numerical systems in this paper, we assume that the constraints are $\ell$-bit-representable w.r.t the numerical system as well. For all three settings, the main objects of our study are the realizability problems.

\begin{definition}[$\ell$-bit-representable Nash Equilibrium Realizability Problem]
    Given a game $\mathbb{G}$, a number $\ell$ represented in unary, and a constraint set $C$, decide if there exists a Nash equilibrium strategy profile $\pi = \langle \pi_1 \ldots \pi_m \rangle $ such that each $\pi_i \in \pi$ is memoryless, $\ell$-bit-representable, and satisfies $C$.  
\end{definition}

\begin{definition}[$\varepsilon$-Nash Equilibrium Realizability Problem]
        Given a game $\mathbb{G}$, a number $\ell$ represented in unary,a number $\varepsilon$ represented in binary, and a constraint set $C$, decide if there exists a $\varepsilon$-equilibrium strategy profile $\pi = \langle \pi_1 \ldots \pi_m \rangle $ such that each $\pi_i \in \pi$ is memoryless, $\ell$-bit-representable, and satisfies $C$.   
\end{definition}

Going forward, we drop the ``$\varepsilon$-Nash'' naming convention, referring to $\varepsilon$-Nash equilibria as simply $\varepsilon$-equilibria. This makes the name shorter and more distinct from ``Nash equilibria''.

\section{Prior Work}\label{sec:stateoftheart}

While the focus on $\ell$-bit-representability in this paper is novel, the Nash-equilibrium realizability and $\varepsilon$-equilibrium realizability problems are significant areas of research in both the game theory and the theoretical computer science literature.  In this section, we provide a brief but thorough survey of known results, focusing on those applicable to our setting, i.e., the one in which each player has a reachability goal. This also serves as a literature review for a somewhat disorganized body of work. 

Before continuing with a survey of the results, it should be noted that \emph{normal-form} games, also called \emph{strategic-form} games~\cite{Osborne1994}, in which each reward is non-negative, can be easily encoded into our reachability goal-setting. This can be done by normalizing the rewards (for example, by dividing by the largest reward) to fall between 0 and 1 and then using simple probabilistic gadgets to encode payoff for players. For example, giving a player a payoff of $\frac{1}{3}$ can be encoded by a state that transitions to the goal of the player with probability $\frac{1}{3}$ and to a non-goal state with probability $\frac{2}{3}$. These two states can then deterministically transition to the same successor state, where a new payoff mechanism for another player can be implemented.  

Nash equilibria may sometimes need distributions that use irrational probabilities~\cite{Nash51}, even in the normal-form non-negative reward setting~\cite{BiloMavronicolas2014}. This means that the restriction to $\ell$-bit-representability is a meaningful one, as the use of rational probabilities may disqualify the certain Nash equilibria. More recent work~\cite{irradicalprobsNash} has shown Nash equilibria in normal-form non-negative reward games may even require \emph{irradical} probabilities, meaning numbers that are algebraic but do not admit a representation using radicals. It should also be noted that in~\cite{BiloMavronicolas2014}, the authors studied decision problems related to the numerical representation of probabilities in Nash equilibria in normal-form games, such as determining whether a game admits a Nash equilibrium that utilizes only rational probabilities and determining whether a game admits a Nash equilibrium that utilizes at least one irrational probability.

For \emph{two}-player games in which both players have a reachability goal (mirroring our setting), it has been proven that an $\varepsilon$-equilibrium exists for every $\varepsilon >0$~\cite{vieille2000a,vieille2000b} (in fact, the result is established for more general games with limit-average payoff goals). In~\cite{chatterjee-majumdar-jurdzinski-csl2004}, it is claimed that there is a two-player game in which both players have a reachability goal but no Nash equilibrium exists. This claim is unfortunately flawed, as the example game considered there \emph{does} admit a pure strategy Nash equilibrium, as is pointed out in~\cite{DBLP:conf/fsttcs/BouyerMS14}. Therefore, it is unknown whether Nash equilibria always exist in two-player games in which both players have a reachability goal. It is further claimed in~\cite{chatterjee-majumdar-jurdzinski-csl2004} that every game with three or more players admits an $\varepsilon$-equilibrium for every $\varepsilon > 0$. This claim is also unfortunately flawed, as demonstrated in~\cite{EteYan08}. Therefore, the existence of $\varepsilon$-equilibria in three-or-more player games is also an open question. 

In~\cite{nostationaryNEkuipers}, the authors demonstrate a 3-player \emph{turn-based} game that admits no memoryless Nash equilibrium. This example is encoded into the reachability goal setting in~\cite{ummelsthesis} (Proposition 3.13), demonstrating that considering only memoryless strategies is also a meaningful restriction. This has motivated various other papers to examine the bounded-memory setting, see~\cite{boundedpartialinformationconcurrency} for one such recent example. The recent works in~\cite{recurisveetr,stationaryepsinturnbased} should also be noted, as they study turn-based models that are otherwise similar to the models considered in this paper. Of particular relevance is their focus on realizing memoryless Nash and $\varepsilon$-equilibria, though it should be noted that~\cite{stationaryepsinturnbased} deals with a \emph{promise} version of the realizability problem. One of the most relevant papers to this work is~\cite{etrstatne}, which makes the claim that determining whether payoff-constrained memoryless Nash equilibria exist in concurrent mean-payoff games belongs to $\exists\mathbb{R}$ (the existential theory of the reals). We re-examine this question in Section~\ref{sec:noell}, where we remove the $\ell$-bit-representable restriction, as although the $\exists\mathbb{R}$ upper bound is correct as a result, the proof is incorrect as written and must be repaired. We discuss this point and review other related results (namely,~\cite{stationaryepsinturnbased,etrstatne,recurisveetr}) in more detail in Section~\ref{sec:noell}.

In summary, not everything is known about $2$-player games in which each player has a reachability goal, and relatively little is known about $n$-player games ($n \geq 3$). From the results that are known, we can deduce that both the memoryless restriction and $\ell$-bit-representability represent meaningful restrictions on the Nash equilibrium and $\varepsilon$-equilibrium realizability problems.

\section{General Constructions and Techniques}\label{sec:constructions}

The techniques presented to establish the upper bound in this Section form the algorithmic core of the NP upper bounds for the $\ell$-bit-realizability problems for both the Nash and $\varepsilon$-equilibria in $\mathbb{Q}$ and $\mathbb{Q}(r_1 \ldots r_m)$ in Section~\ref{sec:Q} and Section~\ref{sec:Qext}. The Markov Chain and Markov Decision Process analysis of Nash and $\varepsilon$-equilibria has been utilized in the literature before, for example~\cite{etrstatne}, but our algorithmic framework uniquely focuses on numerical analysis, such as controlling the size of the numbers considered using our $\ell$-bit bounds.

\begin{enumerate}
    \item A memoryless $\ell$-bit-representable strategy profile $\pi$ is non-deterministically guessed. This uses a polynomial amount of non-determinism, as the strategy profile $\pi$ is memoryless and $\ell$-bit-representable (where $\ell$ is provided in unary in the input of the realizability problems).
    \item In order to verify that $\pi$ satisfies either the Nash equilibrium or $\varepsilon$-equilibrium condition, two probabilities must be compared for each player. The first probability, which we call the \emph{compliance} probability, represents the expected payoff that a Player~$i$ receives when all players follow $\pi$ without deviating. The second probability, which we call the \emph{deviation} probability, represents the maximum expected payoff that Player~$i$ can achieve through a unilateral deviation, i.e., Player~$i$ changes their strategy but all other players follow $\pi$ without deviation.
    \item In order to compute the compliance probability $\rho_i(\mathbb{G},\pi)$, we construct the Markov Chain $\mathbb{G} \times \pi$ -- the cross-product of the game $\mathbb{G}$ and the memoryless strategy profile $\pi$ --  that computes the expected payoff of the players when $\pi$ is followed without deviation.
    \item In order to compute the deviation probability for a specific Player~$i$, we construct the Markov Decision Process $\mathbb{G}_i \times \pi$, which allows Player~$i$ to change their strategy and deviate while all other players follow $\pi$ without deviation. Solving this MDP then yields the deviation probability for Player~$i$.
    \item The Nash equilibrium condition can then be verified by directly comparing the compliance and deviation probabilities for each player (if the deviation probability is strictly larger, then $\pi$ is not a Nash equilibrium). In order to do this comparison, numerical bounds must be established that bound the number of bits that must be compared. The $\varepsilon$-equilibrium condition can then be verified by comparing the sum of the compliance probability and $\varepsilon$ with the deviation probability for each player. These same numerical bound techniques can be applied to the constraint set as well, as we can then check that $\forall i. l_i \leq \rho_i(\mathbb{G},\pi) \leq u_i$. 
\end{enumerate}

We now detail the two main constructions, the Markov Chain $\mathbb{G} \times \pi$ and the Markov Decision Processes $\mathbb{G}_i \times \pi$.

\subsection{The Markov Chain $\mathbb{G} \times \pi$}\label{sec:MCsection}

As mentioned before, the Markov Chain $\mathbb{G} \times \pi$ is used to compute the \emph{compliance probability} $\rho_i(\mathbb{G},\pi)$ for each player, which is the expected payoff they receive when all players follow $\pi$ without deviation. The Markov Chain $\mathbb{G} \times \pi$ is given by the triple $\langle Q, q_0 , E , (R_i)_{i \in \Pi} \rangle$

The state space of $\mathbb{G} \times \pi$ is $Q$, the set of states in $\mathbb{G}$, with the initial state $q_0$ preserved. For the sake of convenience, we carry over the set of reachability goals $(R_i)_{i \in \Pi}$ from $\mathbb{G}$. The most significant part of this construction is the matrix $E$, where $E(q_1, q_2)$ represents the probability that $q_1$ transitions to $q_2$ in $\mathbb{G} \times \pi$. This probability is given by:

$$E(q_1 , q_2) = \sum_{r \in  R_{\delta}(q_1)} \mathbb{P}(\pi(q_1) \sim r) \cdot \mathbb{P}(\delta(q_1, r) \sim q_2)$$

where the notation $\mathbb{P}(d \sim c)$ is the probability that the distribution $d$ assigns to the element (or group of elements) $c$. Intuitively, we iterate over $R_{\delta}(q_1)$, the set of rows in the compressible table $\delta$ associated with $q_1$. For each row, we determine the probability that $\pi$ outputs a joint action that matches the form of $r$ (which we call a \emph{pattern}), denoted by $\pi(q_1) \sim r$. This is relatively easy to compute given the compressible table format, as the probability that $\pi$ outputs a tuple of the form $a{\ast} {\ast} {\ast} b {\ast}$ in a 5-player game, for example,  is simply the probability that the first player outputs $a$ (which is determined by the output of $\pi_1$ at $q_1$) multiplied by the probability that the fourth player outputs $b$ (which is determined by the output of  $\pi_4$ at $q_1$). This probability is then multiplied by the probability that $\delta$ transitions from $q_1$ to $q_2$ upon reading such a joint action, which involves reading the distribution assigned by $\delta$. %This only involves a lookup in $\delta$, as we are quantifying over the rows of $\delta$ to begin with.  
The final row of $R_\delta(q_1)$, which corresponds to all tuples that do not match a listed pattern, is then computed as the difference of $1$ and the sum of the other computed probabilities.

It is critical to note that although the computational complexity of computing $E(q_1,q_2)$ is critical to our analysis, we do not explore it in this section, which is solely focused on establishing general constructions. We leave the complexity-theoretic analysis to the latter sections,  as this allows us to fix the numerical system we are working in. 

The payoff $\rho_i(\mathbb{G},\pi)$ is then the probability that the set $R_i$ is eventually seen in $\mathbb{G} \times \pi$.  For a Player~$i$, this is given by the hitting probability of the set $R_i$ from the initial state $q_0$. It now remains to consider the \emph{deviation probability}, which is the maximum payoff a player can receive when they deviate but all other players still follow $\pi$.

\subsection{The Markov Decision Process $\mathbb{G}_i \times \pi$}\label{sec:MDPsection}

To compute the deviation probability for a Player~$i$, we construct the Markov Decision Process (MDP) $\mathbb{G}_i \times \pi =  \langle Q, q_0, h, A, (R_i)_{i \in \Pi} \rangle $. Conceptually, the only difference between this MDP and $\mathbb{G} \times \pi$ is that in the MDP, we allow a single Player~$i$ to choose their action independently from the recommendation of $\pi$. This allows them to choose the strategy that achieves the maximum expected payoff when all other players commit to following $\pi$. This maximum expected payoff is precisely the deviation probability.

Thus, we specify that the sole player making decisions in the MDP -- the stand-in for Player~$i$ in $\mathbb{G}$ -- has an action set of $A$. The new transition function $h$ now takes as input a state $ q_1$ \emph{and} an action $\overline{a} \in A$ and returns a distribution over possible successor states in $Q$. Therefore, we wish to compute terms of the form $h(q_1, \overline{a}, q_2 )$, the probability that $q_1$ transitions to $q_2$ in $\mathbb{G}_i \times \pi$ when the action $\overline{a} \in A$ is chosen.  

This probability $h(q_1, \overline{a}, q_2 )$ is then given by 

$$h(q_1, \overline{a}, q_2 ) =  \sum_{\substack{r \in R_{\delta}(q_1)\\ r[i] = \overline{a} \vee {\ast}}} \mathbb{P}(\pi(q_1) \sim r[-i]) \cdot \mathbb{P} (\delta( q_1, r) \sim q_2)$$

where $r[-i]$ represents the joint action pattern $r$ with the $i$-th element projected out. Although the formulation may look complex, the idea is relatively simple. Since Player~$i$ has chosen $\overline{a}$ regardless of the output of $\pi$ and this choice is independent of the other players' action choices, we iterate over $R_{\delta}(q_1)$ to find row patterns where the $i$-th element is $\overline{a}$ or ${\ast}$ (e.g. the index condition $r[i] = \overline{a} \vee {\ast}$ in the sum). We then calculate the probability that the non-$i$ elements of $r$ (e.g. $r[-i]$) were output by $\pi$. Since $\pi$ is a tuple of memoryless strategies, this can be done component-wise as in Section~\ref{sec:MCsection}, i.e. $$\prod_{\substack{j \in \Pi \\ j \not = i}} \mathbb{P}(\pi_j(q_1) \sim r[j])$$

the product of the probabilities that the action output by $\pi_j$ matches the $j$-th character of the pattern, skipping over the $i$-th index since we know that is fixed to be $\overline{a}$. Note that this is exactly what we do in the Markov Chain case, but without skipping over the $i$-th index.

With these two constructions in hand, we move on to tackling the different cases created by our choices of numerical system and $\ell$, our bit-size restriction on rational numbers. It is important to recall that when we talk about $\ell$-representability, we say that the probabilities output by $\delta$ are $t$-representable (see Remark~\ref{transitionistrepresentable}). This is due to the fact that the transition table $\delta$ is provided as part of the input, and we derive the parameter $t$ from the representation of $\delta$ in the input.  

\section{$\ell$-bit-representable in $\mathbb{Q}$}\label{sec:Q}

In this section, we consider memoryless strategies that output rational probabilities that can be represented using at most $\ell$ bits. As described in Section~\ref{sec:constructions}, our analysis begins by nondeterministically guessing a strategy profile $\langle \pi_1 \ldots \pi_m \rangle$ by guessing its constituent strategies $\pi_i$, each of which can be represented in polynomial space with respect to the input of $\mathbb{G},k$, and $\ell$. 

The strategy profile $\pi$ outputs joint actions component-wise by considering the output of each $\pi_i$. Thus, it is important to reason about the size of the probabilities that $\pi$ assigns to an element of $R_q(\delta)$. Since each $\pi_i$ is $\ell$-bit-representable, this product, when represented as a rational number, can be represented using at most $2m\ell = O(m\ell)$ bits. This can be seen by noting that the product of $m$ $\ell$-sized integers can be represented using at most $m\ell$ bits, and applying that bound to both the numerators and denominators of the probabilities in $\tau_i$. This bound is polynomial in the size of the input since both $m$ and $\ell$ are provided in unary. Therefore, the probabilities output by $\pi$ have polynomial size w.r.t the input.

The first step is constructing the Markov Chain $\mathbb{G} \times \pi$. In the Markov Chain $\mathbb{G} \times \pi$, the edge relation $E(q_1,q_2)$ was computed through the expression $ \sum_{r \in  R_{\delta}(q_1)} \mathbb{P}(\pi(q_1) \sim r) \cdot \mathbb{P}(\delta(q_1, r) \sim q_2)$. This is a sum with at most $|R_{\delta}(q_1)|$ summands. Each summand is a product of an output probability of $\pi$, which, as described previously, can be represented using at most $2m\ell$ bits (in the worst-case scenario where there are no ${\ast}$ characters), and a probability output by the transition function $\delta$, which, as described in Remark~\ref{transitionistrepresentable}, can be represented using at most $t$ bits.  A single product of the form $\mathbb{P}(\pi(q_1) \sim r) \cdot \mathbb{P}(\delta(q_1, r) \sim q_2)$ can therefore be represented using $O(m\ell t)$ bits.

Since there are at most $|R_{\delta}(q_1)|$ summands,  $E(q_1,q_2)$ can be represented using at most $O(\log(|R_{\delta}(q_1)|m\ell t))$ bits, which is polynomial in the size of the input.    Furthermore, since the number of rows associated with $q_1$ in $\delta$ is represented in unary in the input (as rows in $\delta$), $E(q_1,q_2)$ can be computed in polynomial time w.r.t the size of $\mathbb{G}$, $k$, and $\ell$.

\begin{lemma}\label{lem:k=1-Q-MCpolynomialrep}
    In the $\ell$-bit-representable over $\mathbb{Q}$ setting, the Markov Chain $\mathbb{G}\times \pi$ can be represented in polynomial space w.r.t the size of the input of $\mathbb{G}$ and $\ell$.
\end{lemma}

By an entirely similar argument, the Markov Decision Process $\mathbb{G}_i \times \pi$ can be constructed in polynomial time from the inputs of $\mathbb{G},k$, and $\ell$. The only difference is that instead of an edge relation $E$, we have a transition function $h$. Although the choice of action from Player~$i$ influences the computation of the transition function, an element that was not present in $E$, the computations and reasoning patterns are largely the same -- each value $h(q_1 ,\overline{a}, q_2)$ is a polynomial-size probability --  a sum over the set $R_{\delta}(q_1)$ of rows of a product of $O(k)$ $\ell$-bit-representable probabilities.

\begin{lemma}\label{lem:k=1-Q-MDPpolynomialrep}
    In the $\ell$-bit-representable over $\mathbb{Q}$ setting, the Markov Decision Process $\mathbb{G}_i \times \pi$ can be represented in polynomial space w.r.t the size of the input of $\mathbb{G}$ and $\ell$.
\end{lemma}

With both $\mathbb{G} \times \pi$ and $\mathbb{G}_i \times \pi$ constructed, we can determine whether Player~$i$ has a profitable deviation from $\pi$ in $\mathbb{G}$ by comparing the optimal hitting probability of $R_i$ from $q_0$ that Player~$i$ can obtain in $\mathbb{G}_i \times \pi$ to the hitting probability of $R_i$ from $q_0$ obtained in $\mathbb{G} \times \pi$. We start with $\mathbb{G}_i \times \pi$. Since $\mathbb{G}_i \times \pi$ is a Markov Decision Process with a reachability goal, one of the strategies that realizes the optimal hitting probability in $\mathbb{G} \times \pi$ must be a \emph{policy}, a strategy that is both deterministic and memoryless~\cite{Puterman94}. This policy can be obtained through linear programming methods, which, in general, run in weakly polynomial time~\cite{Puterman94}. Here, however, since we are given the parameter $\ell$ explicitly in unary and all probabilities are of size $\mathrm{poly}(\ell)$, we obtain a polynomial-time algorithm. Although linear programming methods may not directly return a policy, it is a relatively straightforward matter to extract an optimal policy from a linear programming solution~\cite{Puterman94}. The policy itself, being memoryless and deterministic, can be represented in polynomial space. Once this optimal policy is obtained, fixing the behavior of Player~$i$ to this policy in $\mathbb{G}_i \times \pi$ creates a Markov Chain $\hat{\mathbb{G}} \times \pi$ with transition matrix $E^*$ that can be represented in polynomial space w.r.t the size of $\mathbb{G},k$ and $\ell$ -- this is the same argument as Lemma~\ref{lem:k=1-Q-MCpolynomialrep}.  The hitting probability can then be solved by inverting the new transition matrix $E^*$ (which may require polynomial-time pre-processing, see e.g.~\cite{Puterman94,RRV26a}), and since $E^*$ has a polynomial-size representation, the final hitting probability has a polynomial-size representation (applying Cor 3.2a of~\cite{schrijver1998} to show a polynomial-size representation of the inverse). The exact same logic applies to computation in $\mathbb{G} \times \pi$, simply skipping the step where a policy was computed.

\begin{lemma}\label{lem:k=1-Q-hitting-prob-polynomial-representation}
        In the $\ell$-bit-representable over $\mathbb{Q}$ setting, both the compliance probability and the deviation probability can be computed in polynomial time (and therefore, have polynomial space representations).
\end{lemma}
%\myv{Where did we discuss the Markov-Chain-analysis algorithm?} \sr{We didn't, but I added in that inverting the matrix may require processing and references to Puterman and our previous paper that does more or less this explicitly in this setting. Going through the complete pruning process will make a long paper even longer, and we do not need the careful details from before that show we can do this pre-processing in LOGSPACE -- all we are looking for is polynomial time. Computing an MC hitting probability in polynomial time is uncontroversial and well-known, even if people gloss over the details. }

The final step is to then compare the compliance probabilities to the deviation probabilities and the $\ell$-bit-representable constraints. As shown in Lemma~\ref{lem:k=1-Q-hitting-prob-polynomial-representation}, the compliance probability can be written as $\frac{n_c}{d_c}$ ($n$ for numerator, $d$ for denominator), where both $n_c$ and $d_c$ have polynomial size w.r.t. $\mathbb{G},k,$ and $\ell$. By the same logic, the deviation probability can be written as $\frac{n_d}{d_d}$. These two probabilities can then be compared by cross-multiplying: $d_d \cdot n_c$ versus $d_c \cdot n_d$. Both products have polynomial size and can therefore be directly compared in polynomial time. The reasoning for the constraint comparisons is entirely similar. This completes our NP algorithm; the strategy profile $\pi$ can be nondeterministically guessed and then deterministically verified in polynomial time.

\begin{theorem}\label{thm:NE-Q-memoryless-NP-upper}
        The $\ell$-bit-representable Nash Equilibrium Realizability problem over $\mathbb{Q}$ belongs to NP. 
\end{theorem}

The exact same algorithm applies to $\varepsilon$-equilibria, in which the sum of the compliance probability and $\varepsilon$ can be compared to the deviation probability. Since $\varepsilon \in \mathbb{Q}$ is part of the input, it is a straightforward matter to add it to the compliance probability and get a polynomial-size bound on the sum. This straightforward modification also applies to the multiplicative $\varepsilon$ definition, in which $\rho_i(\mathbb{G},\pi)$  is multiplied by $(1 + \epsilon)$ instead of directly summed with $\varepsilon$. Note that $\varepsilon$ does not factor into the constraint sets in either definition.

\begin{theorem}\label{thm:eNE-Q-memoryless-NP-upper}
        The $\ell$-bit-representable $\varepsilon$-Equilibrium Realizability problem over $\mathbb{Q}$ belongs to NP. 
\end{theorem}

\section{$\ell$-bit-representable in $\mathbb{Q}(r_1 \ldots r_m)$}\label{sec:Qext}

In the previous section, Section~\ref{sec:Q}, we used $\ell$-realizability to refer to numbers in $\mathbb{Q}$ that can be represented using at most $\ell$ bits. As mentioned in Section~\ref{sec:stateoftheart}, Nash equilibria may require the use of irrational~\cite{Nash51} or even irradical~\cite{irradicalprobsNash} probabilities. Therefore, restricting the numbers considered in the Nash equilibrium and the $\varepsilon$-equilibrium realizability problem to elements of $\mathbb{Q}$ may disqualify the only equilibria that exist in certain games.

In this section, we show how to extend the analysis in Section~\ref{sec:Q} to consider elements of $\mathbb{Q}(r_1 \ldots r_n)$. This is the field that the results from extending $\mathbb{Q}$ by radicals $r_i \in \mathbb{R}$ of the form $r_i = \sqrt[a]{b}$ for $a \in \mathbb{N}$, $b \in \mathbb{Q}$. This gives us an interesting perspective on the realizability problem -- by specifying \emph{which} {\bf radical-algebraic} real irrational numbers we are looking for, we can still retain the NP upper bounds of Section~\ref{sec:Q}. This result has a greater impact on the $\varepsilon$-realizability problem, which we explain later on, in Section~\ref{sec:discussion}. 

It is important to note that for a radical $r_i$, extending $\mathbb{Q}$ by $r_i$ results in a \emph{radical extension} that results in a field $\mathbb{Q}(r_i)$~\cite{algebra}. This property does not apply in general to all real numbers; for example, $\mathbb{Q}[\pi]$ is a ring, but not a field (thus, we cannot use the notation $\mathbb{Q}(\pi)$). It is also important to distinguish the generators of the extension from the basis, as this distinction plays a central role in our computational analysis. For example, consider the field $\mathbb{Q}(\sqrt{2},\sqrt{3})$, which is $\mathbb{Q}$ extended by the radicals $\sqrt{2}$ and $\sqrt{3}$. In this way, there are two generators of the extension, but an element of $\mathbb{Q}(\sqrt{2},\sqrt{3})$ is actually a 4-tuple of the form $a + b\sqrt{2} + c\sqrt{3} + d\sqrt{6}$ for $a,b,c,d \in \mathbb{Q}$, which corresponds to the basis of $\{1, \sqrt{2}, \sqrt{3}, \sqrt{6}\}$. 

\begin{remark}\label{rem:basis-not-generator}
    We assume that the field $\mathbb{Q}(r_1 \ldots r_m)$ is represented by its \emph{basis} elements, not by its generators. In the worst-case scenario, the number of basis elements can be exponential in the number of generators. In other scenarios, generators can be redundant (such as $\sqrt{2}$ and $\sqrt{8} = 2\sqrt{2}$). Therefore, we are assuming that the list of basis elements was prepared in an algebraic pre-computation step, and that we are fed the basis as input.  
\end{remark}

Remark~\ref{rem:basis-not-generator} gives us an important property. First, note that the number $d$ of elements of $\mathbb{Q}$ used to represent an element of $\mathbb{Q}(r_1 \ldots r_n)$ is part of the input, and therefore, we do not have to deal with exponentially large $d$-tuples of numbers. This exponential-tuple issue would potentially be an issue if we represented $\mathbb{Q}(r_1 \ldots r_m)$ by its generators $r_1 \ldots r_m$. Furthermore, by specifying the basis, we also specify the \emph{degree} of the extension~\cite{algebra} (the number of elements in an independent basis), which we make use of in our calculations. Given a basis $\{1, b_1, b_2 \ldots b_{d-1} \}$, we say that the number $a_0 +a_1 b_1 +a_2 b_2 + \ldots + a_{d-1} b_{d-1} $ is $\ell$-representable iff each $a_i \in \mathbb{Q}$ is $\ell$-representable. We now show that the critical polynomial bounds that held when multiplying over $\mathbb{Q}$ hold when multiplying over $\mathbb{Q}(r_1 \ldots r_n)$. 

\begin{observation}\label{obs:multiplying}
When multiplying two numbers in $\mathbb{Q}(r_1 \ldots r_n)$, an extension of degree $d$ , $a_0 +a_1 b_1 +a_2 b_2 + \ldots + a_{d-1} b_{d-1}$ and $a'_0 +a'_1 b_1 +a'_2 b_2 + \ldots + a'_{d-1} b_{d-1}$, we perform $O(d^2)$ multiplications, and then reduce to basis elements. Since the basis is specified in our input, $d$ is provided in unary in the input, and so $d^2$ is polynomial w.r.t the input. Therefore, the result of the multiplication, $p_0 +p_1 b_1 +p_2 b_2 + \ldots + p_{d-1} b_{d-1}$ can be computed in polynomial time and has a polynomial size representation w.r.t the input of the original two numbers. 
\end{observation}

This critical observation allows us to essentially re-create the constructions given in Section~\ref{sec:Q}. Intuitively, in Section~\ref{sec:Q}, we dealt with numbers that were at most $\ell$ bits in length, and now we are dealing with $d$-tuples (where $d$ is implicitly provided in unary in the input as the length of the basis)  of numbers that are at most $\ell$ bits in length. The most interesting case is that of division, which involves computing an inverse. This can be done by inverting a polynomial-size system of equations, which is also solvable in polynomial time. Crucially, since the polynomial bounds on the complexity of individual arithmetic operations hold, we get analogs of Lemma~\ref{lem:k=1-Q-MCpolynomialrep} and Lemma~\ref{lem:k=1-Q-MDPpolynomialrep}. 

\begin{lemma}\label{lem:Qext-all-representation}
In the $\ell$-bit-representable  setting over $\mathbb{Q}(r_1 \ldots r_m)$, both the Markov Chain $\mathbb{G} \times \pi$ and the Markov Decision Processes $\mathbb{G}_i \times \pi$ can be constructed in polynomial time and have polynomial-size representations with respect to $\mathbb{G},\ell$ and the input field-extension basis $\{1,b_1 \ldots b_{d-1}\}$ (which constitutes a basis of the radical extension $\mathbb{Q}(r_1 \ldots r_m)$).
\end{lemma}

As before, we can solve for the compliance probability in $\mathbb{G} \times \pi$ in polynomial time. As opposed to a single rational number, the solution is an element of $\mathbb{Q}(r_1 \ldots r_m)$ of the form $a_0 + a_1b_1 + \ldots a_{d-1}b_{d-1}$. 

\begin{lemma}\label{lem:Qext-compliance}
   In the $\ell$-bit-representable  setting over $\mathbb{Q}(r_1 \ldots r_m)$,  the compliance probability can be computed in polynomial time (and therefore, has a polynomial space representation).
\end{lemma}

\begin{proof}
%\myv{The sketch does not stand on its own. More context is needed, which is that of MDP analysis}\sr{This isn't MDP analysis -- it's just matrix operations to compute hitting probabilities for the compliance probability.}
%\myv{I meant MC analysis. Please add some context.} \sr{added some context.}
Just as in Section~\ref{sec:Q}, computing the compliance probability boils down to analyzing the Markov Chain $\mathbb{G} \times \pi$. Ultimately, the proof comes down to the observation that we are working in an $|Q|d \times |Q|d$ matrix, as opposed to an $|Q| \times |Q|$ matrix as in Lemma~\ref{lem:k=1-Q-hitting-prob-polynomial-representation}. This is because each element of the transition matrix can be represented as a $d \times d$ matrix~\cite{algebra}, as opposed to the singletons used in the analysis in Section~\ref{sec:Q}. By leveraging the fact that $d$ is implicitly provided in unary in the input, $|Q|d$ is still polynomial in the size of the input, meaning that the previous arguments largely hold with one exception -- multiplication between basis elements is done symbolically.  For example, we would not create a decimal expansion of $\sqrt{2} \times \sqrt{3}$, but rather represent it as $\sqrt{6}$. Throughout the course of our operations, it may be the case that we perform other actions on the basis elements, such as raising $\sqrt{3}$, a basis element, to some power $p$ (or multiplying many basis elements together). Because we are only ever doing a polynomial number of operations, this $p$ must be polynomial, and since each operation increases the size of representation additively (Observation~\ref{obs:multiplying}), the results of these symbolic calculations can also be represented using polynomially many bits. 
\end{proof}
% \jfr{Is the lemma easy to establish ? How do you bound the bit size representation of the values that multiply the elements of the basis in the computed probabilities ?} \sr{Should I explicitly prove this ? The basic idea is that we are not working with an $n$ by $n$ matrix, but rather an $nd$ by $nd$ matrix.}

Next, we must be able to compute the deviation probability. The missing ingredient is the ability to compare two elements of $\mathbb{Q}(r_1 \ldots r_m)$: even if we have two elements $x = x_0 + x_1b_1 \ldots x_{d-1}b_{d-1}$, $y = y_0 + y_1b_1 \ldots y_{d-1}b_{d-1}$ that have polynomial-sized representations with respect to $\mathbb{G}, \ell$ and the basis $\{1, b_1 \ldots b_{d-1}\}$, we must be able to decide if $x \geq y$. This allows us to solve an MDP in the $\mathbb{Q}(r_1 \ldots r_m)$ setting, as it allows us to use linear programming methods, through the same logic employed in Section~\ref{sec:Q}. Furthermore, we need to ultimately compare the compliance probabilities with the deviation probabilities and payoff constraints. Unlike Section~\ref{sec:Q}, such comparisons cannot be done directly through cross-multiplication. Note that the bounds on the representations of numbers prevent a direct comparison to the Square-Root-Sum problem~\cite{complexityzoo}.

\begin{theorem}\label{thm:compare-Qext}
    Given two elements of $\mathbb{Q}(r_1 \ldots r_m)$, $x = x_0 + x_1b_1 \ldots x_{d-1}b_{d-1}$, $y = y_0 + y_1b_1 \ldots y_{d-1}b_{d-1}$ such that each $x_i,y_i$ is a rational number that has a polynomial size representation w.r.t the input of $\mathbb{G}, \ell$ and the basis $B = \{1, b_1, \ldots b_{d-1}\}$, it can decided if $x \geq y$ in polynomial time.
\end{theorem}

\begin{proof}

    The basic idea behind this proof is to enumerate the decimal expansions $x$ and $y$ bit by bit. Although $x$ and $y$ have infinite representations as decimals, the main idea of the proof is that 
    %\myv{Explain what is meant by "polynomial-size bounds".} \sr{repeated the meaning}
    since $x$ and $y$ obey polynomial-size bounds (i.e., have a representation that is polynomial w.r.t $\mathbb{G},\ell,$ and $|B|$), 
    we only need to compare polynomially many bits of their decimal expansions. Therefore, even though radicals like $\sqrt{2}$ (to use one specific example) have an infinite decimal representation, we are able to efficiently compare $\mathrm{poly}(\mathbb{G},\ell,|B|)$-representable elements of $\mathbb{Q}(\sqrt{2})$ by enumerating only polynomially many bits of their decimal expansions. 
    
    The first step is to consider the difference $z = (x_0-y_0) + (x_1-y_1)b_1 + \ldots (x_{d-1}-y_{d-1}) b_{d-1}$ with the ultimate goal being to decide $z \geq 0$. In order for $z$ to be equal to $0$, we must have $\forall i \in \{0 \ldots d-1\}. x_i = y_i $, and we can test for this in polynomial time. Otherwise, we can conclude that $z \not = 0$, and we return to deciding $z > 0$.    Since the RHS of the inequality is $0$, we can clear the denominators of $z$ by multiplying through by the common denominators of all $x_i - y_i$, This involves at most a polynomial number of multiplications by polynomial-sized numbers, and so, the new representation of $z = z_0 + z_1b_1 \ldots z_{d-1}b_{d-1}$ can be computed in polynomial time w.r.t $\mathbb{G}, \ell$ and the basis. It is important to note that $z$ is now an element of the ring $\mathbb{Z}[r_1 \ldots r_m]$, as we now have $z_i \in \mathbb{Z}$ for all $i \in \{0 \ldots d-1\}$.

    The ring $\mathbb{Z}[r_1 \ldots r_m]$ admits 
    %\myv{Ring "norm" shows up here out of the blue.}\sr{Rewrote so its a little less jarring, but I'm not sure what else to do -- we define what the norm is here itself, so it seems we did what we could.}
    a \emph{norm} $N : \mathbb{Z}^d \rightarrow \mathbb{Z}$, 
    a homogeneous multivariate polynomial of degree $d$ with the property that for all nonzero $a$, $|N(a)| \geq 1$~\cite{algebra}. The norm can be evaluated symbolically in polynomial time by multiplying conjugates~\cite{algebra}, which can be generated symbolically since all $r_i$ are radical expressions (see e.g. Example 4.1 in~\cite{adlerbeling}). Previously, we mentioned that $z$ could be viewed as a $d$-tuple of integers $[z_0 \ldots z_{d-1}]$, and so the next step is to apply the norm $N$ to $z$. 

    Since the norm is a polynomial of degree $d$ and $d$ is given in unary in the input, $N(z)$ has a polynomial size representation and is therefore at most singly exponential in value w.r.t. the input. Therefore, if the size of the input is $|i|$, we get $|\frac{N(z)}{z}| \leq 2^{p(|i|)}$ for some polynomial $p(\cdot)$. Since $|N(z)| \geq 1$, we get $|z| \geq 2^{-p(|i|)}$. This is the statement we need, as it lets us test the sign of $z$ by expanding a polynomial number of bits. Once we can decide whether $z > 0$ or $z < 0$, this lets us determine whether $x > y$ or $x < y$ (recalling that we already accounted for the $x=y$ case earlier).   
    
\end{proof}

To assist readers unfamiliar with radical extensions and norms, we provide a brief example illustrating how these norm inequalities operate in a simple field.

\begin{example}
    We work in the field $\mathbb{Q}(\sqrt{2})$. Take a number $y = a + b\sqrt{2}$, with $a$ and $b$ $\ell$-bit-representable rational numbers. After clearing denominators, we have $y \in \mathbb{Z}[\sqrt{2}]$, and, more directly, $a,b \in \mathbb{Z}$. The norm of the ring $\mathbb{Z}[\sqrt{2}]$ is $[a,b] \mapsto a^2 - 2b^2 = (a - b\sqrt{2})( a+ b\sqrt{2})$. Therefore, we have $$ |a + b\sqrt{2}| = \frac{|a^2 -2b^2|}{|a - b\sqrt{2}|} \implies a,b \:\mathrm{ nonzero } \iff |a + b\sqrt{2}| \geq \frac{1}{|a - b\sqrt{2}|}$$
    Here, the replacement of $|a^2 -2b^2|$ by $1$ is justified by the fact that since $a^2 -2b^2$ is the norm, we have $|a^2 - 2b^2| \geq 1$ for all non-zero $a,b$. Since we have bounds on the size of $a$ and $b$ ($\ell$-bit-representable), we can get a bound on the size (and therefore value) of $|a - b\sqrt{2}|$. This bound can then be used to get a bound on $|a + b\sqrt{2}|$ through the above inequality.  In general, for a degree $d$ extension, we get a polynomial of degree $d-1$ in the denominator, and since $d$ is implicitly provided in unary in the input, we get the required size bounds. This can be seen by noting that the size of a term $n^d$ is given by the logarithm, which is $d\log(n)$. In this specific case, $\mathbb{Q}(\sqrt{2})$, $d = 2$, and so we have a linear expression in the denominator.
\end{example}

Theorem~\ref{thm:compare-Qext} fills in our missing ingredient, as it allows us to compare polynomial-sized elements of $\mathbb{Q}(r_1 \ldots r_m)$ in polynomial time. This lets us compute the deviation probability, as solving linear programs (to solve the MDPs $\mathbb{G}_i \times \pi$) requires comparison operations~\cite{adlerbeling,Puterman94} in addition to the standard arithmetic operations discussed in Observation~\ref{obs:multiplying}.
%\myv{Not clear how you go from the  theorem to the lemma.}\sr{Solving an LP requires comparisons, so you need to be able to use $\geq$ in addition to standard arithmetic operations. Added a little more of an explanation in there.}
\begin{lemma}\label{lem:Qext-deviation}
   In the $\ell$-bit-representable over $\mathbb{Q}(r_1 \ldots r_m)$ setting,  the deviation probability can be computed in polynomial time (and therefore, has a polynomial space representation).
\end{lemma}

The final step is then to compare the compliance probabilities, deviation probabilities, and payoff constraints in polynomial time. Since all of these numbers have polynomial size w.r.t the input (Lemmas~\ref{lem:Qext-compliance} and~\ref{lem:Qext-deviation}, and the payoff constraints are part of the input itself), this follows directly from applying Theorem~\ref{thm:compare-Qext}.

This means that for a given strategy profile $\pi$ (which, we recall, was originally non-deterministically chosen, see Section~\ref{sec:constructions}), we can verify if $\pi$ is a Nash equilibrium in polynomial time, giving us an NP upper bound for our realizability problem.
%\myv{Why not in polynomial time?}\sr{Because we started by guessing the memoryless strategy. We guess in nondeterministic polynomial time, and then verify in polynomial time.}\myv{Add that context.} \sr{Added.}
\begin{theorem}\label{thm:NE-Q-ext-NP-upper}
    The $\ell$-bit-representable-Nash Equilibrium Realizability problem over $\mathbb{Q}(r_1 \ldots r_m)$ belongs to NP when the basis of $\mathbb{Q}(r_1 \ldots r_m)$ is provided in the input.
\end{theorem}

These comparisons still hold in the $\varepsilon$ setting, as, given a compliance probability $c = c_0 + c_1b_1 + \ldots + c_{d-1}b_{d-1}$ both the additive definition $c_{\varepsilon} = (c_0 + \varepsilon) + c_1b_1 + \ldots + c_{d-1}b_{d-1}$ and the multiplicative definition $c_{\varepsilon} = (1+ \varepsilon) c_0  +  (1+ \varepsilon)c_1b_1 + \ldots + (1+ \varepsilon)c_{d-1}b_{d-1}$ result in numbers that have polynomial-sized representations, and, therefore, ones we can compare using Theorem~\ref{thm:compare-Qext}, recalling that $\varepsilon$ is an $\ell$-representable rational number provided in the input. 

\begin{theorem}\label{thm:eNE-Q-ext-NP-upper}
        The $\ell$-bit-representable $\varepsilon$-Equilibrium Realizability problem over $\mathbb{Q}(r_1 \ldots r_m)$ belongs to NP when the basis of $\mathbb{Q}(r_1 \ldots r_m)$ is provided in the input.
\end{theorem}

\section{General Lower Bound for $\mathbb{Q}$ and $\mathbb{Q}(r_1\ldots r_m)$}\label{sec:lowerbound}

\begin{figure}
\centering
\includegraphics[width=0.8\textwidth]{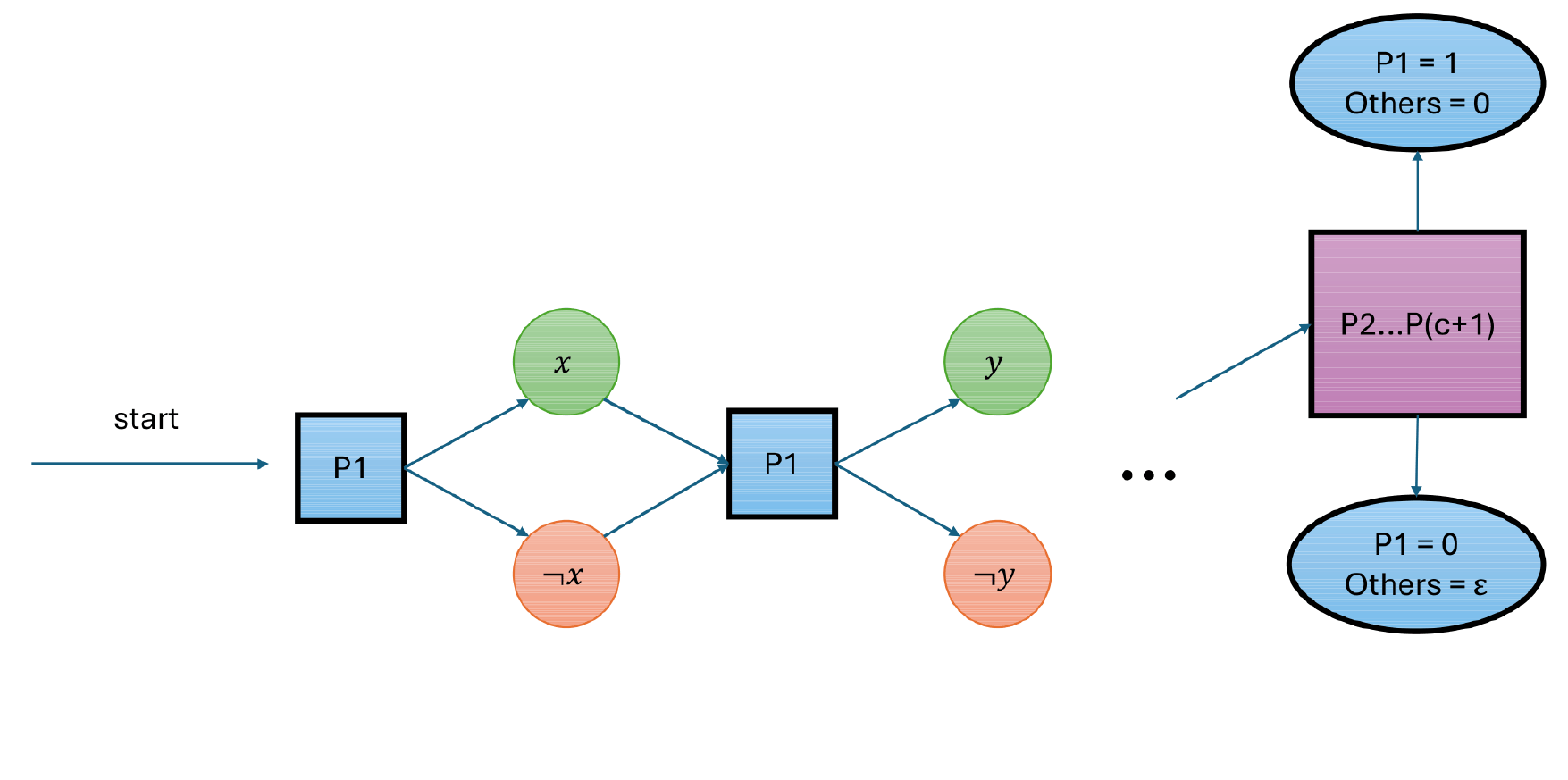}
\caption{The basic schematics for the game used in the NP-hard reduction from 3SAT.}
\label{fig:lowerboundgame}
\end{figure}

Figure~\ref{fig:lowerboundgame} provides the basic framework of our NP-hardness reductions from the 3CNF problem. For a 3CNF formula $\phi$ with $v$ variables and $k$ clauses, we create the game $\mathbb{G}_{\phi}$ depicted in Figure~\ref{fig:lowerboundgame} with $O(v)$ vertices and $k+1$ players. This type of reduction is similar to previous work in the literature, such as~\cite{BBMU15}. As this reductional framework has appeared in the literature before, we focus on a high-level description of how it differs in our setting, as opposed to reiterating already-established low-level details.  

The basic idea is that there is one more player in $\mathbb{G}_{\phi}$ than the number of clauses in $\phi$.  This extra player is Player $1$. For the rest of the players, Player $i$ is ``assigned'' to $C_{i-1}$, the $i-1$st CNF clause in $\phi$.  At the start of the game, Player $1$ sequentially creates an assignment to the $v$ variables by choosing to visit either a positive or a negative form of a variable (e.g. $x$ vs $\neg x$). Exactly one must be visited for each variable. The goal of Player $i$ for $i \not = 1$ then corresponds to the literals in $C_{i-1}$ -- if $C_{i-1}$ is $(x \vee y \vee \neg z)$, for example, then Player $i$'s reachability goal is $\{x,y,\neg z\}$. At the end of Player $1$'s assignment, all other players with clause goals get to vote (corresponding to the purple state), and the resulting transition depends on whether all ``clause'' players vote positively or not. These voting dynamics are succinctly represented by the compressible table format, using one row to capture the transition when all clause players vote positively and then the extra row to capture all other cases (i.e. at least one player votes negatively). If they all vote positively, then no further clause goals are visited, and Player $1$ gets a payoff of $1$. Otherwise, if there is even one dissenter amongst the clause players, Player $1$ gets $0$, and every clause player visits a mechanism that reaches their goal with probability $\varepsilon$.  Here, $\varepsilon$ represents either the input constant $\varepsilon$ from the specification of an $\varepsilon$-equilibrium or an arbitrary non-zero constant for the Nash equilibrium.

Correctness can be seen by asking for a constrained Nash or $\varepsilon$ equilibrium in which Player $1$ gets a payoff of $1$. There is only one outcome that guarantees that -- Player $1$ successfully navigates through the variable assignment and all clause players vote to give him a payoff of $1$. The only way for all clause players to be incentivized to vote positively at the purple vertex is for them to have already previously satisfied their reachability goals in the assignment phase; otherwise, they would be choosing between a payoff of $0$ and $\varepsilon > 0$.  The only way for this to be possible is for there to be a satisfying assignment to $\phi$, giving us NP-hardness for our realizability problems.

\begin{theorem}\label{thm : all-NP-complete}
        The $\ell$-bit-representable-Nash Equilibrium Realizability problem and $\ell$-bit-representable $\varepsilon$-equilibrium Realizability problems over both $\mathbb{Q}$ and $\mathbb{Q}(r_1 \ldots r_m)$ are all NP-hard, and, therefore, NP-complete. 
\end{theorem}

\section{The Unrestricted $\mathbb{R}$ Setting}\label{sec:noell}

In this section, we discard the $\ell$-bit representability assumption. This setting, combined with our memoryless strategy restriction and constrained payoffs, yields a problem that has been studied in the literature previously. Specifically, in~\cite{etrstatne}, this is called the ``StatNE'' problem, and the authors argue that it lies in the complexity class $\exists\mathbb{R}$, the ``existential theory of the reals'', see~\cite{complexityzoo} for more information.  While the precise statement in~\cite{etrstatne} is that the problem lies in PSPACE, this follows directly from the $\exists\mathbb{R}$ upper bound they argue for and the fact that $\exists\mathbb{R} \subseteq \mathrm{PSPACE}$~\cite{etrpspace}. This result is expanded on in~\cite{recurisveetr}, where the authors show an $\exists\mathbb{R}$-hardness result for the same problem in turn-based games. As mentioned in Remark~\ref{rem:turn-based-encoding}, such games can be encoded using our compressible table (Definition~\ref{compressibletabledef}) formulation, implying that the corresponding Nash equilibrium realizability problem is $\exists\mathbb{R}$-complete.

The main contribution of this section is to repair the construction presented in the proof of Theorem 8 in~\cite{etrstatne}, a foundational result in the literature that established an $\exists\mathbb{R}$ upper bound for the stationary Nash equilibrium realizability problem with constrained payoffs and has led to several follow-up works such as~\cite{recurisveetr,stationaryepsinturnbased}, both of which repeat the claims of Theorem 8 in~\cite{etrstatne} as given. With the $\exists\mathbb{R}$ upper bound repaired, it can be combined with Theorem 2 of~\cite{recurisveetr}, which establishes a matching $\exists\mathbb{R}$ lower bound for the final completeness result of this paper. It should be noted that~\cite{stationaryepsinturnbased} deals with a promise version of the realizability problem and is therefore not exactly equivalent, but the point still extends to their Proposition 7, which remakes the construction presented in Theorem 8 of~\cite{etrstatne} from a promise version perspective but admits the same mistake. We now proceed by explaining the mistake in Theorem 8.

\subsection{Prior Construction and Counterexample}

As mentioned before, there exists a 3-player turn-based game with terminal rewards that admits no Nash equilibrium. The game, which was first introduced in~\cite{nostationaryNEkuipers}, is encoded into the state-based terminal reward setting in Proposition 3.13 of~\cite{ummelsthesis} (specifically, Figure 3.3). We denote this game as $G_3$, and adopt the notation of~\cite{ummelsthesis}, referring to the 4 non-terminal states of $G_3$ as $v_0, v_1, v_2$ and $v_3$, with $v_3$ as the uncontrolled start state.

We now construct the three-player turn-based game $G'$, shown in Figure~\ref{fig:counterexamplegame}. It is a relatively simple extension of $G_3$, adding two new states, $s_0$ and $T$. The state $s_0$ is the new initial state, owned by Player $1$.  The new state $T$ is a reachability target for each of the three players.
\begin{figure}
\centering
\includegraphics[width=0.5\textwidth]{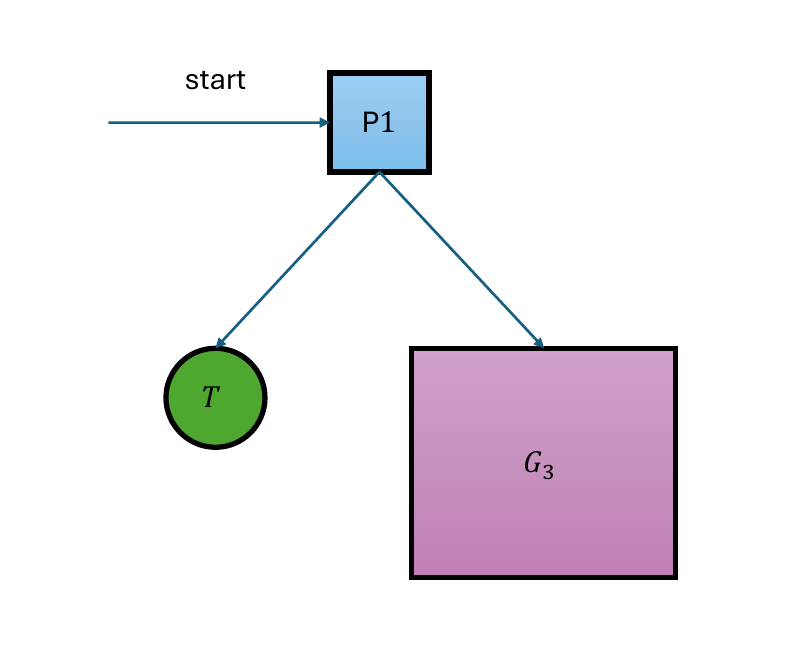}
\caption{The counterexample game $G'$}
\label{fig:counterexamplegame}
\end{figure}

The dynamics of the game are simple. At the start of the game, Player $1$ can either choose to go to $T$ or to go to $v_3$, the start state of $G_3$. These transitions are deterministic, and therefore, the system as a whole is deterministic.  Clearly, there is a memoryless Nash equilibrium in $G'$, which is realized when Player $1$ goes to $T$. This provides the maximal payoff to all players simultaneously, automatically satisfying the no-deviation condition of Nash equilibria. It is important to note that this is the \emph{only} memoryless Nash equilibrium that \emph{can} exist -- if payoff constraints disqualify it, then there is simply no memoryless Nash equilibrium at all. 

Alternatively, Player $1$ can go to $v_3$, and all three players then play $G_3$ a game that is known to have no memoryless Nash equilibrium. It is here that we see a contradiction with the $\exists\mathbb{R}$ sentence constructed in the proof of Theorem 8 of~\cite{etrstatne}. In said $\exists\mathbb{R}$ sentence, the authors use the Bellman optimality equations over all states simultaneously in the constructed MDPs (similar to our construction of $\mathbb{G}_i \times \pi$) to characterize the Nash equilibrium condition. This, however, is too strong a condition. No matter what memoryless strategy profile $\pi$ we consider, the Bellman optimality conditions in $G_3$ will be violated in at least one player's MDP. If this were not the case, then no player would have an incentive to deviate from the candidate strategy profile, meaning that this candidate would be a memoryless Nash equilibrium in $G_3$, contradicting known results.  Therefore, the $\exists\mathbb{R}$ sentence constructed in Theorem 8 of~\cite{etrstatne} must return false, but, as demonstrated earlier, there is a memoryless Nash equilibrium in $G'$. The key point here is that Bellman optimality need not be considered for the states in $G_3$, as the stationary Nash equilibrium strategy profile does not visit $G_3$ with positive probability. In order to repair the result and encode this constraint into our new $\exists\mathbb{R}$ sentence, we therefore develop a reachability predicate for $\exists\mathbb{R}$. 

\subsection{Corrected Construction}

In order to mend the construction of~\cite{etrstatne}, we must restrict the Bellman conditions to only apply to states $R_{\pi}$ that are reachable with positive probability under the guessed strategy profile $\pi$. Therefore, while the \emph{value} of each state is computed just as in~\cite{etrstatne}, the Bellman optimality equations that reason about the value at $q_0$ specifically are only constructed with respect to the set $R_{\pi}$, not the entire state space $Q$. 

\begin{observation}~\cite{SCHWEITZER1976173}
    Given an MDP $M$, a policy $\pi$, and a state $q_0$, the policy achieves the optimal value at $q$ if it satisfies the Bellman optimality equations at all states in $R_{q}^{\pi}$, the set of states in $M$ that are reachable with positive probability in $M$ from $q$ when the policy $\pi$ is followed. In our formulation, $q$ is always the initial state $q_0$ and is therefore not specified. 
\end{observation}

Restricting the Bellman optimality equations to just the set $R_i$ in Figure~\ref{fig:counterexamplegame} means that the Bellman optimality conditions for the states in $G_3$ are correctly ignored, as $G_3$ is not reachable under the guessed Nash equilibrium strategy profile $\pi$, which simply moves to $T$. 

Therefore, in order to mend the construction of~\cite{etrstatne}, we create a polynomial-sized reachability predicate in $\exists\mathbb{R}$. This predicate needs to be created in the $\exists \mathbb{R}$ sentence itself, as we existentially quantify over the strategies in the sentence and can therefore not reference the strategies that determine the reachability sets outside the $\exists \mathbb{R}$ sentence. Once the transition matrix $M$ of $\mathbb{G} \times \pi$ has been computed (using the strategy profile resulting from the existential quantification), which can be done using polynomially many symbolic calculations ( as argued previously in Sections~\ref{sec:Q} and~\ref{sec:Qext} ) we can define a reachability predicate as follows:

$$\mathrm{Reach}(q_1, q_2) := \bigvee_{k \in {0, 1, \ldots |Q|}} M^k[q_1,q_2] > 0$$

by leveraging the fact that, in a Markov Chain (an MDP with a fixed policy $\pi$), if a state $q_2$ is reachable from a state $q_1$ with positive probability, then it is reachable in $k$ steps for some $0 \leq k \leq |Q|$. This allows us to use ``exponentiation'' in our reachability predicate, as $M^k$ can be explicitly expanded as $M \cdot M \cdot M \ldots $ ($k$ times). This reachability predicate, therefore, has polynomial size with respect to $\mathbb{G}$. Note that we have made no attempt to optimize this predicate; our only goal is to show the existence of such a polynomial-sized predicate. This predicate can then be inserted into the sentence in Theorem 8 of~\cite{etrstatne}, which analyzes the MDPs $\mathbb{G}_i \times \pi$ through Bellman conditions on $\mathbb{G} \times \pi$, to correct the result by only considering the Bellman conditions over states that are reachable from the initial state $q_0$.  Clearly, this corrected sentence can then be easily modified to fit the $\varepsilon$-equilibrium concept. 

\begin{theorem}\label{thm:fixed-ETR}
    The Nash Equilibrium and $\varepsilon$-Equilibrium Realizability problems over $\mathbb{R}$ belong to $\exists\mathbb{R}$.
\end{theorem}

The upper bound for the Nash equilibrium can then be combined with the $\exists\mathbb{R}$ lower bound of~\cite{recurisveetr}, which was shown for turn-based systems (see Remark~\ref{rem:turn-based-encoding}), to get the following completeness result.

\begin{theorem}(Lower Bound from~\cite{recurisveetr})
    The Nash Equilibrium Realizability problem over $\mathbb{R}$ is $\exists\mathbb{R}$-complete. 
\end{theorem}

The lower bound construction utilized in~\cite{recurisveetr} does not readily apply to the $\varepsilon$-equilibrium. In fact, in~\cite{stationaryepsinturnbased}, it is conjectured that the $\varepsilon$-equilibrium realizability problem may in fact belong to NP, as only polynomially-sized probabilities may be needed to satisfy the $\varepsilon$-equilibrium condition. We leave this conjecture for future work, leaving the precise complexity of the $\varepsilon$-equilibrium realizability problem over $\mathbb{R}$ open.  Another interesting question concerns the setting without payoff constraints, as the unconstrained payoff $\varepsilon$-equilibrium realizability problem over $\mathbb{R}$ may actually be trivial -- it may be the case that memoryless $\varepsilon$-equilibria always exist. This reasoning does not extend to the Nash equilibrium, as $G_3$ is an example of a game with no memoryless Nash equilibrium.

\section{Discussion}\label{sec:discussion}

The main technical results of this paper are Theorem~\ref{thm : all-NP-complete} and Theorem~\ref{thm:fixed-ETR}. Together, they present the following picture : for the $\ell$-bit-representable cases, both the Nash equilibrium and $\varepsilon$-equilibrium realizability problems are NP-complete. For the unrestricted $\mathbb{R}$ case, the Nash equilibrium problem is $\exists\mathbb{R}$-complete. The corresponding $\varepsilon$-equilibrium belongs to $\exists\mathbb{R}$ as well, but more work must be done to characterize this problem as either $\exists\mathbb{R}$-hard or not.

Taken together, these results present an interesting hierarchy of numerical reasoning. Although it seems that there is something of a ``jump'' between the NP-completeness of the $\mathbb{Q}(r_1 \ldots r_m)$ Nash equilibrium realizability problem and the $\exists\mathbb{R}$-completeness of the $\exists \mathbb{R}$-complete Nash equilibrium realizability problem, the representations we used in this paper paint a subtler story. Specifically, the NP-completeness result for $\mathbb{Q}(r_1 \ldots r_m)$ hinges on the fact that the extension is specified through its basis, not by its generators. If we had shown an NP upper bound for the problem variant in which the extension was specified by its generators, then this would provide a plausible route to showing that NP $=\exists\mathbb{R}$. The high-level idea would be that even the correct generators themselves could be guessed in NP, collapsing $\exists\mathbb{R}$ to NP by showing that the $\exists\mathbb{R}$-complete problem of Nash equilibrium realizability belongs to NP. In order to establish such a claim, however, it must be shown that the $\ell$-bit property could be preserved properly, which we leave for future work. Thus, we have an interesting picture in which increasing the number of generators potentially leads to more meaningful and accurate answers but also potentially blows up the basis, and, therefore, the search space, exponentially. 

As mentioned before, this greater accuracy plays a much more impactful role with $\varepsilon$-equilibria than with Nash equilibria. It is precisely this difference that leads to the uncertain nature of the exact complexity $\varepsilon$-realizability problem over $\mathbb{R}$. Because we are working with a fixed basis, even if we do not pinpoint the specific algebraic numbers needed to exactly solve for an $\varepsilon$-equilibrium, we may be able to get a much better approximation to the exact solution (should it exist) than by just using rational numbers. This can be seen by applying a statement like Drichlet's Approximation Theorem~\cite{hardy75}. Because the bounds of approximation are potentially tightened, this allows for feasible solutions with even very small $\varepsilon$, which, in practice, are interchangeable with exact solutions. The exact analysis of the magnitude of these improvements is outside the scope of this paper and left for future work. 

\section{Conclusion}\label{sec:conc}

In this paper, we studied the computational complexity of the Nash and $\varepsilon$-realizability problems under different numerical restrictions. In addition to the technical results, this paper makes two important contributions to the literature.

First, it introduces the compressible table definition and carefully analyzes the complexity of the Markov Chain and Markov Decision Process constructions that result from the compressible table format. While analyzing multiplayer state-based games using MCs and MDPs is a well-known idea in the literature, the compressible table provides an important bridge between the turn-based games previously studied in works such as~\cite{recurisveetr,etrstatne,stationaryepsinturnbased} and the general concurrent setting. Crucially, the compressible table allows us to maintain the polynomial time bounds for computing transition probabilities that are often taken for granted in turn-based settings by explicitly taking into account the number of rows present in the table. Furthermore, it allows for succinct representations for turn-based games as well, allowing us to build upon previous works. Therefore, the compressible table ``threads the needle'' in a sense, in that it both allows for succinct representation for certain important subclasses of games, such as turn-based games, while not blowing up the complexity of querying such a table beyond PTIME. If, for example, we represented the transition table as a Bayesian Network instead, then while this would allow for succinct representations for a wider class of games, it would also mean that simply querying the table would be at least PP-hard~\cite{bayesianinferencecomplexity}; see~\cite{RRV26a} for further discussion of this point. 

Second, it provides a thorough literature review of an important topic that the literature has veered away from. In part, this is due to the state of the literature,  which makes it difficult for researchers to tell what is known. We provide a through review of the literature in Section~\ref{sec:stateoftheart}, providing a centralized reference for the current state of the art that, to the best of our knowledge, did not exist before. In addition to organizing popular claims and refutations, we correct a further long-standing mistake in the literature in Section~\ref{sec:noell}. This is done with the hope that researchers looking to continue this line of work have a centralized starting point from which to survey known results.

Finally, this brings us to the technical contributions of this paper. Our first contribution was to show that calculations over the bounded $\mathbb{Q}$ or in the fixed-point precision setting could be performed in polynomial time. Therefore, we can guess a strategy profile in non-deterministic polynomial time and verify it in polynomial time. This result aligns well with intuition given the state of the literature, but there are still many technical details to work out in the relatively new concurrent setting. The focus on $\ell$-bit-representability allowed us to conduct precise numerical analysis, a detail that is often taken for granted in the literature.

We then extended this reasoning to field extensions of $\mathbb{Q}$, providing an entirely new setting for analysis. While the NP bounds were maintained, they were done with respect to the explicit representation of the basis. Note that this means the radicals cannot be guessed a priori in non-deterministic polynomial time so easily, as constructing the basis from the guessed radicals could take exponential time. These results were established in order to highlight the interesting role that numerical complexity plays in the realizability problems, as we see a clear tradeoff between the reasoning power of $\mathbb{Q}(r_1 \ldots r_n)$ and the computational complexity of working in such a large field. Furthermore, we extended our numerical analysis techniques to this more complicated group of fields. All relevant problems were then given a matching lower bound, allowing us to characterize them all as NP-complete simultaneously.

After this, we removed the $\ell$-bit-representability restriction completely, electing to work in $\mathbb{R}$. This setting aligned closely with other settings in the literature, but there was still work to be done. Namely, we presented a correction to Theorem 8 in~\cite{etrstatne}, a long-standing result that inspired several follow-up works such as~\cite{recurisveetr}. While we were able to provide tight complexity results for the Nash equilibrium, the $\varepsilon$-equilibrium problem was left open, as we believe numerical analysis plays the most critical role here. This problem represents the most important direction for future work, as a tight complexity result would crystallize the interplay between numerical complexity and the Nash and $\varepsilon$ equilibria solution concepts, the main conceptual point of this paper.

\bibliography{bib}
\bibliographystyle{abbrv}

\end{document}